\documentclass[9pt,a4paper]{article}
\usepackage{epsf, graphicx}
\usepackage{latexsym,amsfonts,amsbsy,amssymb}
\usepackage{amsmath,amsthm}
\usepackage{cite}
\usepackage{cleveref}
\usepackage{indentfirst}

\makeatletter

\@addtoreset{equation}{section} \makeatother
\newtheorem{theorem}{Theorem}[section]
\newtheorem{lemma}{Lemma}[section]
\newtheorem{corollary}{Corollary}[section]
\newtheorem{remark}{Remark}[section]
\newtheorem{example}{Example}[section]
\newtheorem{proposition}{Proposition}[section]

\makeatletter \@addtoreset{equation}{section} \makeatother
\begin{document}

\title{ $\sigma$-self-orthogonal constacyclic codes of length $p^s$ over $\mathbb F_{p^m}+u\mathbb F_{p^m}$ \footnote{
 E-Mail addresses: hwliu@mail.ccnu.edu.cn (H. Liu), jinggeliu@mails.ccnu.edu.cn (J. Liu).}}

\author{Hongwei Liu,~Jingge Liu}

\date{\small
School of Mathematics and Statistics, Central China Normal University,Wuhan, Hubei, 430079, China\\
}
\maketitle


%

\section*{Abstract}
 \addcontentsline{toc}{section}{\protect Abstract} 
 \setcounter{equation}{0} 

In this paper, we study the $\sigma$-self-orthogonality of constacyclic codes of length $p^s$ over the finite commutative chain ring $\mathbb F_{p^m} + u \mathbb F_{p^m}$, where $u^2=0$ and $\sigma$ is a ring automorphism of $\mathbb F_{p^m} + u \mathbb F_{p^m}$.  First, we obtain the structure of $\sigma$-dual code of a $\lambda$-constacyclic code of length $p^s$ over $\mathbb F_{p^m} + u \mathbb F_{p^m}$. Then, the necessary and sufficient conditions for a $\lambda$-constacyclic code to be $\sigma$-self-orthogonal are provided.  In particular, we determine the $\sigma$-self-dual constacyclic codes of length $p^s$ over $\mathbb F_{p^m} + u \mathbb F_{p^m}$.  Finally, we extend the results to constacyclic codes of length $2 p^s$.

\medskip
\noindent{\large\bf Keywords: }\medskip   constacyclic code; repeated-root code; $\sigma$-self-orthogonal code; $\sigma$-self-dual code;  finite commutative chain ring.

\medskip
2010 {\it Mathematics Subject Classification.} \, 94A55, 94B05

\section{Introduction}
The study of constacyclic codes originated in the 1960s.
Berlekamp\cite{r2,r3} introduced the concept of negacyclic codes over finite fields. Constacyclic codes are the natural generalization of cyclic codes which can be technically implemented by shift registers. They have similar algebraic structure to cyclic codes so that they inherit most of the good properties of cyclic codes. The properties of constacyclic codes are easy to analyze so that they can easily be encoded and decoded. Thus, this family of codes is interesting for both theoretical and practical reasons.

Codes over finite rings have received much attention recently after it was proved that some important families of binary non-linear codes are in fact images under a Gray map of linear codes over $\mathbb{Z}_4$ (see, for example, \cite{r4,r5,r6}). If the characteristic of the finite ring is relatively prime to the length of a constacyclic code, we call this code a \emph{simple-root code}; otherwise it is called a \emph{repeated-root code}. Dinh and L\'{o}pez-Permouth\cite{r7} obtained the structure of simple-root cyclic and negacyclic codes of length $n$ and their duals over a finite chain ring and gave necessary and sufficient conditions for the existence of a simple-root cyclic self-dual code over a finite chain ring. Since the decomposition of polynomials over finite rings is not unique, the structure of repeated-root constacyclic codes over finite rings is more complex.
Since 2003, some special classes of repeated-root constacyclic codes over certain finite chain rings have been studied by many authors (see, for example,\cite{b25,b27,r8,r9,r10,r11}).

In 1997, Bachoc\cite{r12} discussed linear codes over $\mathbb F_q+u\mathbb F_q$($\,q=p$ or $p^2\, $, $\,p$ is a prime). This work has aroused the interest of researchers in studying error correcting codes over finite chain rings of the form $\mathbb F_{p^m}+u\mathbb F_{p^m}\, $($\,u^2=0,\,p $ is a prime).
Dinh\cite{r13} studied all constacyclic codes of length $2^s$ over $\mathbb F_{2^m}+u\mathbb F_{2^m}$. The algebraic structure of all constacyclic codes of length $p^s$ and $2p^s$ over the finite commutative chain ring $\mathbb F_{p^m}+u\mathbb F_{p^m}$ was determined in \cite{r14,r15}.

Self-orthogonal codes over finite rings or finite fields are a class of important linear codes which are closely related to combinatorial designs and modular lattices. It has been found that the problem of finding quantum error-correcting codes can be transformed into the problem of finding additive codes over $\mathbb F_{4}$ that are self-orthogonal with respect to a certain trace inner product \cite{a17}. This has caused a great interest in constructing classical self-orthogonal codes. Self-dual codes are a special class of self-orthogonal codes. A self-dual code has the same weight distribution as its dual code. A large number of good codes are self-dual codes. So they have been an important subject in the research of error-correcting codes. As far as we know, there have been very few results concerning self-orthogonal constacyclic codes over $\mathbb F_{p^m}+u\mathbb F_{p^m}$. Recently, all self-dual constacyclic codes of length $p^s$ over the finite commutative chain ring $\mathbb F_{p^m}+u\mathbb F_{p^m}$ and the number of each type of self-dual constacyclic code were established in \cite{r16}. But it is not easy to obtain the the self-orthogonality from self-duality.

Let $R$ be a finite commutative Frobenius ring with an identity and $Aut(R)$ be the ring automorphism group of $R$. Let $\sigma\in Aut(R)$, then $\sigma$ can be extended to a bijective map
\begin{displaymath}
\begin{aligned}
R^n &\longrightarrow R^n,\\
(\,r_0,\,r_1,\,\cdots ,\,r_{n-1}\,)& \longmapsto (\,\sigma(r_0),\,\sigma(r_1),\,\cdots ,\,\sigma(r_{n-1})\,).
\end{aligned}
\end{displaymath}

Given $n$-tuples $\textbf{x}=(\,x_0, \,x_1, \,\cdots , \,x_{n-1}),\, \textbf{y}=(\,y_0,\, y_1,\, \cdots ,\, y_{n-1}\,)\in R^n$, their \emph{$\sigma$-inner product} is defined as
$$\langle \mathbf{x},\mathbf{y}\rangle _{\sigma}=\sum _{i=0}^{n-1}x_i\sigma(y_i)=x_0\sigma(y_0)+x_1\sigma(y_1)+\cdots+x_{n-1}\sigma(y_{n-1}).$$

When $R$ is a finite field $\mathbb F_{p^m}$ of order $p^m$, $\sigma$-inner product is just the usual Euclidean inner product if $\sigma$ is the identity map of $\mathbb F_{p^m}$, $\sigma$-inner product is the Hermitian inner product if $m$ is even and $\sigma$ maps any element $a$ of $\mathbb F_{p^m}$ to $a^{p^{\frac{m}{2}}}$ and $\sigma$-inner product is the Galois inner product \cite{c16} if $\sigma$ maps any element $a$ of $\mathbb F_{p^m}$ to $a^{p^{h}}$ for some integer $0\leq h\leq m-1$.

$\sigma$-inner product over finite commutative Frobenius rings generalizes the Euclidean inner product, the Hermitian inner product and Galois inner product over finite fields.
Two $n$-tuples $\textbf{x}$ and $\textbf{y}$ are called \emph{$\sigma$-orthogonal} if $\langle \mathbf{x},\mathbf{y}\rangle _{\sigma}=0$.
For a code $C$ over $R$, its \emph{$\sigma$-dual code} $C^{\perp_\sigma}$ is defined as
\begin{displaymath}
C^{\perp_\sigma} =\left \{\textbf{x}\,|\langle \mathbf{c},\mathbf{x}\rangle _{\sigma}=0 ,\, \forall \,\textbf{c}\in C\right \}.
\end{displaymath}
A code $C$ is called \emph{$\sigma$-self-orthogonal} if $C\subseteq C^{\perp_\sigma}$, and it is called \emph{$\sigma$-self-dual} if $C=C^{\perp_\sigma}$.

This paper focuses on the $\sigma$-self-orthogonality of constacyclic codes of length $p^s$ over the finite commutative chain ring $\mathbb F_{p^m}+u\mathbb F_{p^m}$.

The remainder of this paper is organized as follows. Preliminary concepts and some known results are given in Section $2$. In Section $3$, we obtain the structure of $\sigma$-dual codes of $\lambda$-constacyclic codes of length $p^s$ over $\mathbb F_{p^m}+u\mathbb F_{p^m}$.
In Section $4$, we provide necessary and sufficient conditions for a $\lambda$-constacyclic code to be $\sigma$-self-orthogonal using the relation between the polynomials of the generating sets of a $\lambda$-constacyclic codes and its $\sigma$-dual code. In particular, we obtain the $\sigma$-self-dual constacyclic codes over $\mathbb F_{p^m}+u\mathbb F_{p^m}$. The results in Section $4$ can be extended to constacyclic codes of length $2p^s$ over $\mathbb F_{p^m}+u\mathbb F_{p^m}$.

\section{Preliminaries}

Let $R$ be a finite commutative Frobenius ring with an identity. We call a nonempty subset $C$ of $R^n$ a \emph{code} of length $n$ over $R$ and the ring $R$ is referred to as the \emph{alphabet} of $C$. If $C$ is an $R$-submodule of $R^n$, then $C$ is said to be \emph{linear}. It is easy to obtain the following proposition.

\begin{proposition}\label{pro1}
Let $C$ be a code of length $n$ over $R$, then

(1) $C^{\perp_\sigma}$ is a linear code over $R$.

(2) $C^{\perp_\sigma}=\sigma^{-1}(C^{\perp})$. Moreover, if $C$ is a linear code, then $\left |C \right |\left | C^{\perp_\sigma} \right |=\left | R \right |^n $.
\end{proposition}
\begin{proof}
(1) For any $r_1,r_2\in R,~\mathbf{y}_i=(y_{i0},y_{i1},\ldots ,y_{i,n-1}) \in C^{\perp_\sigma}, i=1,2,$ and
$\mathbf{x}=(x_0,x_1,\ldots ,x_{n-1})\in C,$ we have
\begin{displaymath}
\begin{split}
\langle \mathbf{x},r_1\mathbf{y}_1+r_2\mathbf{y}_2\rangle _{\sigma}&=\sum _{i=0}^{n-1}x_i\sigma(r_1y_{1i}+r_2y_{2i})=\sigma(r_1)\sum _{i=0}^{n-1}x_i \sigma(y_{1i})+\sigma(r_2)\sum _{i=0}^{n-1}x_i \sigma(y_{2i})\\
 &=\sigma(r_1)\langle \mathbf{x},\mathbf{y}_1\rangle _{\sigma}+\sigma(r_2)\langle \mathbf{x},\mathbf{y}_2\rangle _{\sigma}=0.
\end{split}
\end{displaymath}
Thus, $r_1\mathbf{y}_1+r_2\mathbf{y}_2\in C^{\perp_\sigma}$, which means that $C^{\perp_\sigma}$ is a linear code over $R$.

(2) For any $\mathbf{y}\in R^n$, we have $\mathbf{y}\in C^{\perp_\sigma}$ if and only if $\langle \mathbf{c},\mathbf{y}\rangle _{\sigma}=\langle \mathbf{c},\sigma(\mathbf{y})\rangle=0 ,\, \forall \,\textbf{c}\in C$ if and only if $\sigma(\mathbf{y})\in C^{\perp}$ if and only if $\mathbf{y}\in \sigma^{-1}(C^{\perp})$, implying that $C^{\perp_\sigma}=\sigma^{-1}(C^{\perp})$, where $\langle -,-\rangle$ is the usual Euclidean inner product.

Since $\sigma$ can be extended to a bijective map from $R^n$ to $R^n$, $|C^{\perp_\sigma}|=|\sigma^{-1}(C^{\perp})|=|C^{\perp}|.$ If $C$ is  linear, then $\left |C \right |\left | C^\perp \right |=\left | R \right |^n $. Hence, $\left |C \right |\left | C^{\perp_\sigma} \right |=\left | R \right |^n $.
\end{proof}

For a unit $\lambda$ of $R$, the $\lambda\, $-constacyclic ($\lambda\,$-twisted) shift $\tau _\lambda$ on $R^n$ is the shift
\begin{displaymath}
\tau_\lambda (\,x_0,\, x_1,\, \cdots , \,x_{n-1})=(\,\lambda x_{n-1},\, x_0,\, x_1, \,\cdots ,\, x_{n-2}).
\end{displaymath}
A linear code $C$ is said to be \emph{$\lambda\, $-constacyclic} if $\tau_ \lambda (C)=C$.
The $1$-constacyclic codes are the \emph{cyclic codes} and the $-1$-constacyclic codes are just the \emph{negacyclic codes}.

Let $f(x)$ be a polynomial over $R$ and let $\deg f(x)$ denote the degree of $f(x)$.
Under the standard $R$-module isomorphism
\begin{displaymath}
\begin{aligned}
R^n &\longrightarrow R[x]/\left \langle  x^n-\lambda\right \rangle, \\
(\,c_0,\,c_1,\,\cdots ,\,c_{n-1}\,)& \longmapsto  c_0+c_1x+\cdots+c_{n-1}x^{n-1}+\left \langle  x^n-\lambda\right \rangle,
\end{aligned}
\end{displaymath}
each codeword $\textbf{c}=(\,c_0,\,c_1,\,\cdots ,\,c_{n-1}\,)$ can be  identified with its polynomial representation
$$c(x)=c_0+c_1x+\cdots+c_{n-1}x^{n-1}\in R\left [ x \right ],\, \deg\,c(x) \leqslant n-1,$$
and each $\lambda$-constacyclic code $C$ of length $n$ over $R$ can also be viewed as an ideal of the quotient ring $R[x]/\left \langle  x^n-\lambda\right \rangle.$ In the light of this, the study of $\lambda$-constacyclic codes of length $n$ over $R$ is equivalent to the study of ideals of the quotient ring $R[x]/\left \langle  x^n-\lambda\right \rangle.$
It is easy to prove the following proposition by Proposition \ref{pro1}.
\begin{proposition}\label{pro2}
The $\sigma$-dual code of a $\lambda$-constacyclic code is a $\sigma^{-1}(\lambda^{-1})$-constacyclic code.
\end{proposition}

Let $f(x)=a_0+a_1x+\cdots+a_rx^r\in R\left [ x \right ]$, where $a_r\neq 0$.
Then the polynomial $f^\ast (x)=a_r+a_{r-1}x+a_{r-2}x^2+\cdots+a_0x^r$ is called the \emph{reciprocal polynomial} of $f(x)$.
In fact, $f^\ast (x)$ can also be expressed as $f^\ast (x)=x^rf(\frac{1}{x})$. We can see that if $I$ is an ideal of $R[x]/\left \langle  x^n-\lambda\right \rangle$, then $I^\ast=\left \{  f^\ast(x)|f(x)\in I\right \}$ is an ideal of $R[x]/\left \langle  x^n-\lambda^{-1}\right \rangle$.

Let $I$ be an ideal of $R[x]/\left \langle  x^n-\lambda\right \rangle$. The \emph{annihilator} of $I$ denoted by $\mathcal{A}(I)$ is defined as
\begin{displaymath}
\mathcal{A}(I)=\left \{ g(x)\in R[x]/\left \langle  x^n-\lambda\right \rangle \,|\,f(x)g(x)=0,\, \forall \,f(x)\in I \right \}.
\end{displaymath}
Then $\mathcal{A}(I)$ is also an ideal of $R[x]/\left \langle  x^n-\lambda\right \rangle$. It is clear that if $C$ is a $\lambda$-constacyclic code of length $n$ over $R$, then $C^{\perp}$ is $\mathcal{A}(C)^\ast$ and $C^{\perp_\sigma}$ is $\sigma^{-1}(\mathcal{A}(C)^\ast)$.

Throughout this paper, let $p$ be an odd prime and $s$ be a positive integer. $\mathbb F_{p^m}$ denotes the finite field of order $p^m$, where $m$ is a positive integer. $\mathbb F_{p^m}^\ast $ denotes the multiplicative cyclic group of non-zero elements of $\mathbb F_{p^m}$.
Let $R=\mathbb F_{p^m}+u\mathbb F_{p^m}$, where $u^2=0$.
Then $R$ is a finite commutative chain ring with the unique maximal ideal $\langle u\rangle$, whose ideals are $\langle u^0\rangle=R$, $\langle u\rangle$ and $\langle u^2\rangle=0$. Each element of $R$ can be expressed as $a+ub$, where $a,b\in\mathbb F_{p^m}$. Then element $a+ub$ is a unit of $R$ if and only if $a\neq0$. If $a\neq0$, then $a+ub$ is a square of $R$ if and only if $a$ is a square of $\mathbb F_{p^m}$. The automorphism group of $R$ is given as follows.

\begin{proposition}\label{pro3}(\cite{d20})
For $\theta\in Aut(\mathbb F_{p^m})$ and $\varepsilon \in \mathbb F_{p^m}^\ast$, let
\begin{displaymath}
\begin{aligned}
\Theta_{\theta,\varepsilon} :~~\,R&\longrightarrow \,R, \\
a+ub & \longmapsto  \, \theta(a)+\varepsilon \theta(b).
\end{aligned}
\end{displaymath}
Then $Aut(R)=\{\Theta_{\theta,\varepsilon}~|~\theta\in Aut(\mathbb F_{p^m}),~\varepsilon \in \mathbb F_{p^m}^\ast\}$.
\end{proposition}

In the rest of this paper, let $\sigma \in Aut(R)$. Then $\sigma^{-1} \in Aut(R)$, i.e., $\sigma^{-1}=\Theta_{\theta,\varepsilon}$ for some $\theta\in Aut(\mathbb F_{p^m})$ and $\varepsilon \in \mathbb F_{p^m}^\ast$.

For a code $C$ of length $n$ over $R$, its \emph{torsion} and \emph{residue codes} are defined as follows.
\begin{displaymath}
Tor(C)=\left \{ \textbf{b} \in \mathbb F_{p^m}^n\,|\,u\textbf{b}\in C\right \}, \,\,Res(C)=\left \{ \textbf{a} \in \mathbb F_{p^m}^n\,|\,\exists ~ \textbf{b} \in \mathbb F_{p^m}^n \,~such~\, that~ \,\textbf{a}+u\textbf{b}\in C\right \}.
\end{displaymath}
Then both of them are codes of length $n$ over $\mathbb F_{p^m}$. The reduction modulo $u$ from $C$ to $Res(C)$ is defined as
\begin{displaymath}
\begin{aligned}
\phi :~~\,C&\longrightarrow \,Res(C), \\
\textbf{a}+u\textbf{b} & \longmapsto  \, \textbf{a}.
\end{aligned}
\end{displaymath}
Clearly, $\phi$ is well defined and surjective, with $\ker(\phi )=~uTor(C)$, $\phi (C)=Res(C)$. Therefore, $|C|=|Res(C)|\cdot|Tor(C)|$.

\section{The structure of $\sigma$-dual codes of $\lambda$-constacyclic codes}

The algebraic structure of all $\lambda$-constacyclic codes of length $p^s$ over $\mathbb F_{p^m}+u\mathbb F_{p^m}$ was obtained in \cite{r14}.
The situation of $\lambda$ is divided into two cases separately: (a) $\lambda =\alpha +u\beta$, where $\alpha \, $, $\beta $ are nonzero elements of $\mathbb F_{p^m}$, (b) $\lambda =\gamma $, where $\gamma\, $ is a nonzero element of $\mathbb F_{p^m}$.

First, we consider the case that $\lambda =\alpha +u\beta$, where $\alpha \, $, $\beta $ are nonzero elements of $\mathbb F_{p^m}$.

Let $\mathcal{R}_{\alpha , \beta } =\frac{R[x]}{\left \langle x^{p^s}-(\alpha +u\beta ) \right \rangle}$, then the $(\alpha +u\beta)$-constacyclic codes of length $p^s$ over $R$ are ideals of the ring $\mathcal{R}_{\alpha , \beta }$.
By the division algorithm, there exist nonnegative integers $\alpha _q$, $\alpha _r$ such that $s=\alpha _qm+\alpha _r$, and $0\leqslant \alpha _r\leqslant m-1$.
Let $\alpha _0=\alpha ^{-p^{(\alpha _q+1)m-s}}$. Then $\alpha _0^{p^s}=\alpha ^{-1}$. We have the following conclusions:

\begin{lemma}[\cite{r14}]\label{le1}
In $\mathcal{R}_{\alpha , \beta }$, $\left \langle (\alpha _0x-1)^{p^s} \right \rangle=\left \langle u \right \rangle$.  In particular, $\alpha _0x-1$ is nilpotent in $\mathcal{R}_{\alpha , \beta }$ with nilpotency index $2p^s$. $\mathcal{R}_{\alpha , \beta }$ is a chain ring with ideals that are precisely
\begin{displaymath}
\mathcal{R}_{\alpha , \beta }=\left \langle 1 \right \rangle\supsetneqq\left \langle \alpha _0x-1 \right \rangle\supsetneqq \cdots \supsetneqq\left \langle (\alpha_0x-1)^{2p^s-1} \right \rangle\supsetneqq\left \langle (\alpha_0x-1)^{2p^s} \right \rangle=\left \langle 0 \right \rangle.
\end{displaymath}
$(\alpha +u\beta)$-constacyclic codes of length $p^s$ over $R$ are the ideals $\left \langle (\alpha_0x-1)^i\right \rangle$, $0\leqslant i\leqslant 2p^s\,, $ of the chain  ring $\mathcal{R}_{\alpha , \beta }$. Each code $\left \langle (\alpha_0x-1)^i\right \rangle$ contains $p^{m(2p^s-i)}$ codewords.
\end{lemma}

\begin{theorem}\label{th1}
Let $C$ be an $(\alpha +u\beta)$-constacyclic code of length $p^s$ over $R$ and $C=\left \langle (\alpha_0x-1)^i\right \rangle\subseteq \mathcal{R}_{\alpha , \beta }$ for some $0\leq i\geq 2p^s$, then its $\sigma$-dual code is the $(\theta(\alpha^{-1})-u\varepsilon\theta(\beta \alpha^{-2}))$-constacyclic code $$C^{\perp_\sigma}=\left \langle \big(\theta(\alpha_0^{-1})x-1\big)^{2p^s-i} \right \rangle \subseteq  \mathcal{R}_{\theta(\alpha^{-1}) , -\varepsilon\theta(\beta \alpha^{-2})},$$ which contains $p^{mi}$ codewords.
\end{theorem}
\begin{proof}
It follows immediately from Proposition \ref{pro1} and Theorem 4.3 of \cite{r14}.
\end{proof}

In the following, we discuss the case that $\lambda =\gamma $, where $\gamma\, $ is a nonzero element of $\mathbb F_{p^m}$.

Let $\mathcal{R}_\gamma=\frac{R[x]}{\left \langle x^{p^s}-\gamma \right \rangle}$, then the $\gamma$-constacyclic codes of length $p^s$ over $R$ are ideals of the ring $\mathcal{R}_\gamma$.
By the division algorithm, there exist nonnegative integers $\gamma _q$, $\gamma _r$ such that $s=\gamma _qm+\gamma _r$, and $0\leqslant \gamma _r\leqslant m-1$.
Let $\gamma _0=\gamma ^{-p^{(\gamma _q+1)m-s}}=\gamma ^{-p^{m-\gamma _r}}$. Then $\gamma _0^{p^s}=\gamma ^{-p^{(\gamma _q+1)m}}=\gamma ^{-1}$.

In \cite{r14} the authors studied $\gamma$-constacyclic codes by constructing a one-to-one correspondence between cyclic and $\gamma$-constacyclic codes as follows.

\begin{proposition}[\cite{r14}]\label{pro4}
The map $\Psi:\frac{R[x]}{\left \langle x^{p^s}-1 \right \rangle} \rightarrow \frac{R[x]}{\left \langle x^{p^s}-\gamma \right \rangle}$ given by $f(x)\mapsto f(\gamma _0x)$ is a ring isomorphism. In particular, for $A \subseteq \frac{R[x]}{\left \langle x^{p^s}-1 \right \rangle},\,B \subseteq \frac{R[x]}{\left \langle x^{p^s}-\gamma \right \rangle} $ with $\Psi(A)=B$, then $A$ is an ideal of $\frac{R[x]}{\left \langle x^{p^s}-1 \right \rangle}$ if and only if $B$ is an ideal of $\frac{R[x]}{\left \langle x^{p^s}-\gamma \right \rangle}$. Equivalently, $A$ is a cyclic code of length $p^s$ over $R$ if and only if $B$ is a $\gamma$-constacyclic code of length $p^s$ over $R$.
\end{proposition}

Now, using the isomorphism $\Psi$, the results about cyclic code of length $p^s$ over $R$ in \cite{r14} can
be applied to corresponding $\gamma$-constacyclic codes of length $p^s$ over $R$.

\begin{lemma}[\cite{r14}]\label{le2}
$\gamma _0x-1$ is nilpotent in $\mathcal{R}_\gamma$ with nilpotency index $p^s$. The ring $\mathcal{R}_\gamma$ is a local ring with the maximal ideal $\left \langle \,u, \,\gamma _0x-1\, \right \rangle$, but it is not a chain ring.
\end{lemma}

Then all the $\gamma$-constacyclic codes of length $p^s$ over $R$ , i.e., ideals of the local ring $\mathcal{R}_\gamma$, are classified into four types, as follows.

\begin{theorem}[\cite{r14}]\label{th2}
$\gamma$-constacyclic codes of length $p^s$ over $R$ , i.e., ideals of the local ring $\mathcal{R}_\gamma$, are:
\begin{itemize}
  \item Type 1 (trivial ideals): $\left \langle \,0\, \right \rangle$, $\left \langle \,1\, \right \rangle$.
  \item Type 2 (principal ideals with nonmonic polynomial generators): $\left \langle \,u(\gamma _0x-1)^i\, \right \rangle$,  where $0\leqslant i\leqslant p^s-1$.
  \item Type 3 (principal ideals with monic polynomial generators): $$\left \langle \,(\gamma _0x-1)^i+u(\gamma _0x-1)^th(x)\, \right \rangle, $$ where $1\leqslant i\leqslant p^s-1$, and either $h(x)$ is $0$ or $h(x)$ is a unit where it can be represented as $h(x)=\sum _jh_j(\gamma _0x-1)^j$, $h_j\in \mathbb F_{p^m}$ and $h_0\neq 0$.
  \item Type 4 (nonprincipal ideals): $\left \langle \,(\gamma _0x-1)^i+u(\gamma _0x-1)^th(x),  u(\gamma _0x-1)^\omega\, \right \rangle$, where $1\leqslant i\leqslant p^s-1$, $\omega < T$, where $T$ is the smallest integer such that $u(\gamma _0x-1)^T\in \left \langle \,(\gamma _0x-1)^i+u(\gamma _0x-1)^th(x) \,\right \rangle$, with $h(x)$ as in Type 3, and $\deg h(x)\leqslant \omega -t-1$.
\end{itemize}
\end{theorem}

For $\gamma$-constacyclic codes of Type 4 in Theorem \ref{th2}, the number $T$ plays a
very important role. We now determine $T$ for each code $C=\left \langle \,(\gamma _0x-1)^i+u(\gamma _0x-1)^th(x),  u(\gamma _0x-1)^\omega\, \right \rangle$.

\begin{proposition}[\cite{r14}]\label{pro5}
Let $T$ be the smallest integer such that $$u(\gamma _0x-1)^T\in \left \langle (\gamma _0x-1)^i+u(\gamma _0x-1)^th(x) \right \rangle\,, $$ then
\begin{displaymath}
T = \left\{ \begin{array}{ll}
i,  & \textrm{if $h(x)=0$};\\
\min\left \{ \, i\,, \,p^s-i+t\right \},  & \textrm{if $h(x)\neq 0 $}.
\end{array} \right.
\end{displaymath}
\end{proposition}

We now compute the size of each $\gamma$-constacyclic code $C$. By the definitions of $Tor(C)$ and $Res(C)$ and $|C|=|Res(C)|\cdot|Tor(C)|$, we have the following result.

\begin{theorem}[\cite{r14}]\label{th3}
$C$ is $\gamma$-constacyclic codes of length $p^s$ over $R$, as classified in Theorem~3.2.
Then the number of codewords $n_C$ of $C$ is determined as follows.
\begin{itemize}
  \item If $C=\left \langle \,0 \,\right \rangle$, then $n_C=1$.
  \item If $C=\left \langle \,1 \,\right \rangle$, then $n_C=p^{2mp^s}$.
  \item If $C=\left \langle \,u(\gamma _0x-1)^i\, \right \rangle$, where $0\leqslant i\leqslant p^s-1$, then $n_C=p^{m(p^s-i)}$.
  \item If $C=\left \langle \,(\gamma _0x-1)^i\, \right \rangle$, where $1\leqslant i\leqslant p^s-1$, then $n_C=p^{2m(p^s-i)}$.
  \item If $C=\left \langle \,(\gamma _0x-1)^i+u(\gamma _0x-1)^th(x)\, \right \rangle$, where $1\leqslant i\leqslant p^s-1$, $0\leqslant t\leqslant i-1$, and $h(x)$ is a unit, then
 \begin{displaymath}
n_C = \left\{ \begin{array}{ll}
p^{2m(p^s-i)},  & \textrm{if $1\leq i \leq \frac{p^{s}+t}{2}$};\\
p^{m(p^s-t)},  & \textrm{if $\frac{p^{s}+t}{2} <i \leq p^{s}-1$}.
\end{array} \right.
\end{displaymath}
  \item If $C=\left \langle \,(\gamma _0x-1)^i+u(\gamma _0x-1)^th(x),  u(\gamma _0x-1)^\kappa \, \right \rangle$, where $1\leqslant i\leqslant p^s-1$, $0\leqslant t< i$, $h(x)$ is $0$ or a unit, and
\begin{displaymath}
\kappa < T = \left\{ \begin{array}{ll}
i,  & \textrm{if\, $h(x)=0$};\\
\min\left \{  i, p^s-i+t\right \},  & \textrm{if\, $h(x)\neq 0 $}.
\end{array} \right.
\end{displaymath}
then $n_C=p^{m(2p^s-i-\kappa )}$.
\end{itemize}
\end{theorem}

\begin{theorem}\label{th4}
Let $C$ be a $\gamma$-constacyclic code of length $p^s$ over $\mathbb F_{p^m}+u\mathbb F_{p^m}$.
\begin{itemize}
  \item If $C=\left \langle 0 \right \rangle$, then $C^{\perp_\sigma}=\mathcal{R}_1$. If $C=\left \langle 1 \right \rangle=\mathcal{R}_1$, then $C^{\perp_\sigma}=\left \langle 0 \right \rangle$.
  \item If $C=\left \langle u(\gamma _0x-1)^i \right \rangle$, where $0\leqslant i\leqslant p^s-1$, then $C^{\perp_\sigma}=\left \langle \big(\theta(\gamma _0^{-1})x-1\big)^{p^s-i}, u \right \rangle$.
  \item If $C=\left \langle (\gamma _0x-1)^i+u(\gamma _0x-1)^th(x) \right \rangle$, where $1\leqslant i\leqslant p^s-1$, $0\leqslant t< i$, $h(x)$ is $0$ is a unit which can be represented as $h(x)=\sum _jh_j(\gamma _0x-1)^j$, $h_j\in \mathbb F_{p^m}$ and $h_0\neq 0$.

(1)~If $h(x)=0$, then $C^{\perp_\sigma}=\left \langle \big(\theta(\gamma _0^{-1})x-1\big)^{p^s-i} \right \rangle$.

(2)~If $h(x)$ is a unit, and $1\leq i \leq \frac{p^{s}+t}{2}$, then $C^{\perp_\sigma}=\left \langle a(x) \right \rangle,$ where
$$a(x)= \big(\theta(\gamma _0^{-1})x-1\big)^{p^s-i}-u\big(\theta(\gamma _0^{-1})x-1\big)^{p^s+t-2i}\sum _{j=0}^{i-t-1}\varepsilon \theta(h_j)(-1)^{j+t-i}\big(\theta(\gamma _0^{-1})x-1\big)^jx^{i-j-t}.$$

(3)~If $h(x)$is a unit, and $\frac{p^{s}+t}{2} <i \leq p^{s}-1$, then $C^{\perp_\sigma}=\left \langle b(x),u\big(\theta(\gamma _0^{-1})x-1\big)^{p^s-i} \right \rangle,$ where

$$b(x)= \big(\theta(\gamma _0^{-1})x-1\big)^{i-t}-u\sum _{j=0}^{p^s-i-1}\varepsilon \theta(h_j)(-1)^{j+t-i}\big(\theta(\gamma _0^{-1})x-1\big)^jx^{i-j-t}.$$

  \item If $C=\left \langle (\gamma _0x-1)^i+u(\gamma _0x-1)^th(x),  u(\gamma _0x-1)^\omega \right \rangle\,, $  where $1\leqslant i\leqslant p^s-1$, $\omega < T$, $T$ is the smallest integer such that $u(\gamma _0x-1)^T\in \left \langle (\gamma _0x-1)^i+u(\gamma _0x-1)^th(x) \right \rangle$, $h(x)$ is defined as Type 3 in Theorem~3.2, and $\deg{h(x)}\leqslant \omega -t-1$.

(1)~If $h(x)=0$, then $C^{\perp_\sigma}=\left \langle \big(\theta(\gamma _0^{-1})x-1\big)^{p^s-\omega }, u \big(\theta(\gamma _0^{-1})x-1\big)^{p^s-i }\right \rangle$.

(2)~If $h(x)$ is a unit, then $C^{\perp_\sigma}=\left \langle d(x),u\big(\theta(\gamma _0^{-1})x-1\big)^{p^s-i} \right \rangle,$ where
$$d(x)= \big(\theta(\gamma _0^{-1})x-1\big)^{p^s-\omega }-u\big(\theta(\gamma _0^{-1})x-1\big)^{p^s-i-\omega+t }\sum _{j=0}^{\omega -t-1}\varepsilon \theta(h_j)(-1)^{j+t-i}\big(\theta(\gamma _0^{-1})x-1\big)^jx^{i-j-t}.$$
\end{itemize}
\end{theorem}
\begin{proof}
It follows immediately from Proposition \ref{pro1} and Theorem 4.7, Theorem 4.8, Theorem 4.9 of \cite{r16}.
\end{proof}

\section{$\sigma$-self-orthogonal $\lambda$-constacyclic codes of length $p^s$}
In this section, we study the $\sigma$-self-orthogonality of the $\lambda$-constacyclic codes of length $p^s$ over $\mathbb F_{p^m}+u\mathbb F_{p^m}$. We provide necessary and sufficient conditions for a $\lambda$-constacyclic code to be $\sigma$-self-orthogonal using the relation between the polynomials of the generating sets of a $\lambda$-constacyclic codes and its $\sigma$-dual code. In particular, we get $\sigma$-self-duality of the constacyclic codes over $\mathbb F_{p^m}+u\mathbb F_{p^m}$ from their $\sigma$-self-orthogonality.

Consider the code $\left \langle u \right \rangle=\,uR^n=\left \{ \,u\textbf{c}\,|\,\textbf{c}\in R^n\,  \right \}$ of length $n$ over $R$. Clearly, for any unit $\lambda$ of $R$, $uR^n$ is $\gamma$-constacyclic code of length $n$ over $\mathbb F_{p^m}+u\mathbb F_{p^m}$. It is also the ideal of $\frac{R[x]}{\left \langle x^n-\lambda  \right \rangle}$ generated by $u$. $\left \langle u \right \rangle$ can also denote this ideal.
\subsection{$\sigma$-self-orthogonal ($\alpha +u\beta$)-constacyclic codes}
$\lambda =\alpha +u\beta$, where $\alpha \, $, $\beta $ are nonzero elements of $\mathbb F_{p^m}$.  $\mathcal{R}_{\alpha , \beta } =\frac{R[x]}{\left \langle x^{p^s}-(\alpha +u\beta ) \right \rangle}\,,\,\alpha _0^{p^s}=\alpha ^{-1}\,,\,\alpha _0 \in \mathbb F_{p^m}$.

\begin{theorem}\label{th5}
Let $C$ be an ($\alpha +u\beta$)-constacyclic code of length $p^s$ over $\mathbb F_{p^m}+u\mathbb F_{p^m}$ and $C=\left \langle (\alpha_0x-1)^i\right \rangle\subseteq \mathcal{R}_{\alpha , \beta }$ for some $0\leqslant i\leqslant 2p^s$, then $C$ is $\sigma$-self-orthogonal if and only if $p^s\leqslant i\leqslant 2p^s$.
\end{theorem}
\begin{proof}
By Theorem \ref{th1}, $C^{\perp_\sigma}=\left \langle \big(\theta(\alpha_0^{-1})x-1\big)^{2p^s-i} \right \rangle\subseteq  \mathcal{R}_{\theta(\alpha^{-1}) , -\varepsilon\theta(\beta \alpha^{-2})}$.

Necessity. If $C$ is $\sigma$-self-orthogonal, then $C\subseteq C^{\perp_\sigma} $, which yields $|C|\leqslant |C^{\perp_\sigma}|\,.$ It follows from Theorem \ref{th1} that $|C|=p^{m(2p^s-i)}\,,\, |C^{\perp_\sigma}| =p^{mi}$. That means that $p^{m(2p^s-i)}\leqslant p^{mi}$, i.e., $i\geqslant p^s $. Since $0\leqslant i\leqslant 2p^s$, we have $p^s\leqslant i\leqslant 2p^s$.

Sufficiency. From $i\geqslant p^s $ we have $ p^s\geqslant 2 p^s-i $. Therefore,
\begin{displaymath}
\begin{split}
C&=\left \langle (\alpha_0x-1)^i\right \rangle=\left \langle (\alpha_0x-1)^{p^s}(\alpha_0x-1)^{i-p^s}\right \rangle=\left \langle u(\alpha_0x-1)^{i-p^s} \right \rangle\\
&\subseteq \left \langle u \right \rangle=\left \langle \big(\theta(\alpha_0^{-1})x-1\big)^{p^s} \right \rangle\subseteq \left \langle \big(\theta(\alpha_0^{-1})x-1\big)^{2p^s-i} \right \rangle=C^{\perp_\sigma},
\end{split}
\end{displaymath}
which implies that $C$ is $\sigma$-self-orthogonal.
\end{proof}

\begin{theorem}\label{th6}
$\left \langle u \right \rangle$ is the unique $\sigma$-self-dual ($\alpha +u\beta$)-constacyclic code of length $p^s$ over $\mathbb F_{p^m}+u\mathbb F_{p^m}$.
\end{theorem}
\begin{proof}
Since $\left \langle u \right \rangle=\left \langle (\alpha_0x-1)^{p^s} \right \rangle$, It follows from Theorem \ref{th5} that $\left \langle u \right \rangle$ is $\sigma$-self-orthogonal, i.e., $\left \langle u \right \rangle\subseteq\left \langle u \right \rangle^{\perp_\sigma}$. By Theorem \ref{th1}, $|\left \langle u \right \rangle|=p^{m(2p^s-p^s)}=p^{mp^s}$, $|\left \langle u \right \rangle^{\perp_\sigma}|=p^{mp^s}$, which means that $|\left \langle u \right \rangle|=|\left \langle u \right \rangle^{\perp_\sigma}|$. Hence, $\left \langle u \right \rangle=\left \langle u \right \rangle^{\perp_\sigma}$, i.e., $\left \langle u \right \rangle$ is $\sigma$-self-dual.

Next, we will prove the uniqueness. Assume that $C=\left \langle (\alpha_0x-1)^i \right \rangle\subseteq \mathcal{R}_{\alpha , \beta }$ is $\sigma$-self-dual, then $C=C^{\perp_\sigma}$. That follows that $|C|=|C^{\perp_\sigma}|$. In view of Theorem \ref{th1}, $|C|=p^{m(2p^s-i)}$, $|C^{\perp_\sigma}|=p^{mi}$. This leads to $p^{m(2p^s-i)}=p^{mi}$, i.e., $i=p^s$, implying that $C=\left \langle (\alpha_0x-1)^{p^s} \right \rangle=\left \langle u \right \rangle$.
\end{proof}

\subsection{$\sigma$-self-orthogonal $\gamma$-constacyclic codes}
$\lambda =\gamma $, where $\gamma\, $ is a nonzero element of $\mathbb F_{p^m}$. $\mathcal{R}_\gamma=\frac{R[x]}{\left \langle x^{p^s}-\gamma \right \rangle}$, $\gamma _0^{p^s}=\gamma ^{-1}\,,\,\gamma _0 \in \mathbb F_{p^m}$. Since $\sigma^{-1}(\gamma^{-1})=\theta^{-1}(\gamma^{-1})$, the $\sigma$-dual code of a $\gamma$-constacyclic code is a $\theta^{-1}(\gamma^{-1})$-constacyclic code and $\big(\theta(\gamma _0^{-1})\big)^{p^s}=\theta^{-1}(\gamma^{-1})$

Obviously, we have the following result.

\begin{theorem}\label{th7}
Let $C$ be a $\gamma$-constacyclic code of length $p^s$ over $R$ and $C=\left \langle \,0\, \right \rangle$ or $C=\left \langle \, 1 \,\right \rangle$. Then $C$ is $\sigma$-self-orthogonal if and only if $C=\left \langle \,0\, \right \rangle$.
\end{theorem}

 Clearly, when $C=\left \langle \,0\, \right \rangle$ or $C=\left \langle \,1\, \right \rangle$, $C$ is not $\sigma$-self-dual.

\begin{theorem}\label{th8}
Let $C$ be a $\gamma$-constacyclic code of length $p^s$ over $R$ and $C=\left \langle u(\gamma _0x-1)^i \right \rangle$, where $0\leqslant i\leqslant p^s-1$, then

(1) $C$ is $\sigma$-self-orthogonal.

(2) $C$ is $\sigma$-self-dual if and only if $i=0$, i.e., $C=\left \langle \,u\, \right \rangle$.
\end{theorem}
\begin{proof}
By Theorem \ref{th4}, $C^{\perp_\sigma}=\left \langle \big(\theta(\gamma _0^{-1})x-1\big)^{p^s-i}, u \right \rangle\subseteq \mathcal{R}_{\theta(\gamma^{-1})}.$

(1) Since $C=\left \langle u(\gamma _0x-1)^i \right \rangle\subseteq \left \langle u \right \rangle\subseteq \left \langle \big(\theta(\gamma _0^{-1})x-1\big)^{p^s-i}, u\, \right \rangle=C^{\perp_\sigma}$, $C\subseteq C^{\perp_\sigma}\, $, i.e., $C$ is $\sigma$-self-orthogonal.

(2) Necessity. Because $C$ is $\sigma$-self-dual, i.e., $C = C^{\perp_\sigma} $, $|\,C\,|=|\,C^{\perp_\sigma}\,|$. In view of Theorem \ref{th3}, $|\,C\,|=p^{m(p^s-i)}$, $|\,C^{\perp_\sigma}\,|=p^{m(p^s+i)}$, which implies that $p^{m(p^s-i)}=p^{m(p^s+i)}$. Hence, $i=0$. Then $C=\left \langle \,u\, \right \rangle$.

Sufficiency. If $C=\left \langle \,u\, \right \rangle$, then $C^{\perp_\sigma}=\left \langle \,u\, \right \rangle$. It follows that $C=C^{\perp_\sigma}$, i.e., $C$ is $\sigma$-self-dual.
\end{proof}

\begin{lemma}\label{le3}
Let $C$ be a Type 3 or Type 4 $\gamma$-constacyclic code of length $p^s$ over $R$, as in Theorem \ref{th2}. If $C$ is $\sigma$-self-orthogonal, then $\gamma _0=\theta(\gamma _0^{-1})$.
\end{lemma}
\begin{proof}
Suppose $C$ is a Type 3 or Type 4 $\gamma$-constacyclic code of length $p^s$ over $R$, as in Theorem \ref{th2}. It implies that there exist $0\leqslant i\leqslant p^s-1$ and $g(x)\in\mathbb F_{p^m}[x],\,\deg g(x)\leqslant p^s-1$ such that $c(x)=(\gamma _0x-1)^i+ug(x)\in C$. Let $\textbf{c}$ be the $n$-tuples in $C$ whose polynomial representation is $c(x)$. Then $\textbf{c}\in C\subseteq C^{\perp_\sigma}$. However, $\deg c(x)\leqslant p^s-1$, which means that the polynomial representation of $\textbf{c}\in C^{\perp_\sigma}$  is also $c(x)$, i.e., $c(x)\in C^{\perp_\sigma} \lhd \mathcal{R}_{\theta(\gamma^{-1})}\,.$ Assume that $\gamma _0\neq\theta(\gamma _0^{-1})$, then $\theta(\gamma _0)\gamma _0-1\neq0.$ Thus, $\gamma _0x-1=(\theta(\gamma _0)\gamma _0-1)+\theta(\gamma _0)\gamma _0\big(\theta(\gamma _0^{-1})x-1\big)$ is a unit of $\mathcal{R}_{\theta(\gamma^{-1})}$. Therefore, $C^{\perp_\sigma} = \mathcal{R}_{\theta(\gamma^{-1})}$, which means $C=\left \langle \,0\, \right \rangle$, which contradicts that $C\neq \left \langle \,0\, \right \rangle$. Hence, $\gamma _0=\theta(\gamma _0^{-1})$.
\end{proof}

If $\gamma _0=\theta(\gamma _0^{-1})$, then $\gamma=\theta(\gamma^{-1})=\sigma^{-1}(\gamma^{-1})$. It means that the $\sigma$-dual code of a $\gamma$-constacyclic code of length $p^s$ over $R$ is also a $\gamma$-constacyclic code of length $p^s$ over $R$.

\begin{theorem}\label{th9}
Let $C$ be a $\gamma$-constacyclic code of length $p^s$ over $R$ and $C=\left \langle \,(\gamma _0x-1)^i\, \right \rangle$, where~$1\leqslant i\leqslant p^s-1$, $\gamma _0=\theta(\gamma _0^{-1})$, then $C$ is $\sigma$-self-orthogonal if and only if $i\geqslant \frac{p^s}{2}$.
\end{theorem}
\begin{proof}
By Theorem \ref{th4} and $\gamma _0=\theta(\gamma _0^{-1})$, $C^{\perp_\sigma}=\left \langle \,(\gamma _0x-1)^{p^s-i} \,\right \rangle\subseteq \mathcal{R}_{\gamma}$.

Necessity. Since $C\subseteq C^{\perp_\sigma} $, we have $p^s-i\leqslant i$, i.e., $i\geqslant \frac{p^s}{2}$.

Sufficiency. Since $i\geqslant \frac{p^s}{2}$, $p^s-i\leqslant i$. It means that $C\subseteq C^{\perp_\sigma} $, i.e., $C$ is $\sigma$-self-orthogonal.
\end{proof}

\begin{remark} Under the conditions of Theorem \ref{th9}, there doesn't exist $\sigma$-self-dual code for any given $\sigma \in Aut(R)$. In fact, assume that there exists a $\sigma$-self-dual code $C$, then $C = C^{\perp_\sigma} $. Hence, $|\,C\,|=|\,C^{\perp_\sigma}\,|$. It follows from Theorem \ref{th3} that
$|\,C\,|=p^{2m(p^s-i)}$, $|\,C^{\perp_\sigma}\,|=p^{2mi}$. It implies that $p^{2m(p^s-i)}=p^{2mi}$, i.e., $2i=p^s$, which contradicts that $p$ is an odd prime.
\end{remark}

\begin{theorem}\label{th10}
Let $C$ be a $\gamma$-constacyclic code of length $p^s$ over $R$ and
$$C=\left \langle \,(\gamma _0x-1)^i+u(\gamma _0x-1)^th(x)\, \right \rangle\subseteq \mathcal{R}_{\gamma},$$
where $\gamma _0=\theta(\gamma _0^{-1})$, $h(x)$ is a unit where it can be represented $h(x)=\sum _{j=0}^{i-t-1}h_j(\gamma _0x-1)^j$, $h_j\in \mathbb F_{p^m}$, $h_0\neq 0$ and $1\leq i \leq \frac{p^{s}+t}{2}$. Then $C$ is $\sigma$-self-orthogonal if and only if one of the following holds:

(a) $p^s\leqslant i+t$;

(b) $i+t\leqslant p^s\leqslant 2i$ and $(\gamma_0x-1)^{p^s-i-t}\,|\,(h(x)-h^{'}(x))$, where $$h^{'}(x)=-\sum _{j=0}^{i-t-1}\varepsilon \theta(h_j)(-1)^{j+t-i}( x-1)^jx^{i-j-t}.$$
\end{theorem}

\begin{proof}
By Theorem \ref{th4} and $\gamma _0=\theta(\gamma _0^{-1})$,
\begin{displaymath}
C^{\perp_\sigma}
=\left \langle (\gamma_0x-1)^{p^s-i}-u(\gamma_0x-1)^{p^s+t-2i}\sum _{j=0}^{i-t-1}\varepsilon \theta(h_j)(-1)^{j+t-i}( x-1)^jx^{i-j-t} \right \rangle.
\end{displaymath}
Let
$$h^{'}(x)=-\sum _{j=0}^{i-t-1}\varepsilon \theta(h_j)(-1)^{j+t-i}( x-1)^jx^{i-j-t}.$$
Then $$C^{\perp_\sigma}=\left \langle \,(\gamma_0x-1)^{p^s-i}+u(\gamma_0x-1)^{p^s+t-2i}h^{'}(x)\, \right \rangle \subseteq \,\mathcal{R}_\gamma.$$

Necessity.
Since $C\subseteq C^{\perp_\sigma} $, there exist $f_1(x)+uf_2(x), f_i(x)\in \mathbb{F}_{p^m}[x]\, ,\, i=1, 2\, $ such that
\begin{equation}
\begin{split}
&(\gamma_0x-1)^i+u(\gamma_0x-1)^th(x)\\
 \equiv &\,[f_1(x)+uf_2(x)][(\gamma_0x-1)^{p^s-i}+u(\gamma_0x-1)^{p^s+t-2i}h^{'}(x)]\\
 \equiv &\,(\gamma_0x-1)^{p^s-i}f_1(x)+u[(\gamma_0x-1)^{p^s-i}f_2(x)+(\gamma_0x-1)^{p^s+t-2i}f_1(x)h^{'}(x)](mod(\gamma_0x-1)^{p^s}).\label{e1}
\end{split}
\end{equation}

Next, we will prove that we can find $ f_i(x)\in \mathbb{F}_{p^m}[x]\, ,\, i=1, 2\, $ satisfying Equation (\ref{e1}) and
$$f_1(x)=(\gamma_0x-1)^{2i-p^s}\,,\,\deg f_2(x)\leqslant i-1\,.$$

From Equation (\ref{e1}) we have $(\gamma_0x-1)^i \equiv (\gamma_0x-1)^{p^s-i}f_1(x)(\textrm{mod}(\gamma_0x-1)^{p^s})\,, $ which means that there exists $r(x)\in \mathbb{F}_{p^m}[x]$ such that $$(\gamma_0x-1)^{p^s-i}f_1(x)-(\gamma_0x-1)^i=r(x)(\gamma_0x-1)^{p^s}.$$
Hence,
\begin{displaymath}
\begin{split}
(\gamma_0x-1)^i &=(\gamma_0x-1)^{p^s-i}f_1(x)-r(x)(\gamma_0x-1)^{p^s}\\
 &=(\gamma_0x-1)^{p^s-i}[f_1(x)-r(x)(\gamma_0x-1)^i],
\end{split}
\end{displaymath}
implying that $i\geqslant p^s-i$, i.e., $p^s\leqslant 2i$  and $f_1(x)=(\gamma_0x-1)^{2i-p^s}+(\gamma_0x-1)^ir(x)\,.$ Let
$$f_1^{'}(x)=(\gamma_0x-1)^{2i-p^s}\,, \,f_2^{'}(x)=f_2(x)+(\gamma_0x-1)^tr(x)h^{'}(x), $$
and write $f_2^{'}(x)$ as the form
$f_2^{'}(x)=\sum _jd_j(\gamma_0x-1)^j$, $d_j\in \mathbb F_{p^m}$. Let $$f_2^{''}(x)=\sum _{j-0}^{i-1}d_j(\gamma_0x-1)^j\,.$$ Then
\begin{displaymath}
\begin{split}
&(\gamma_0x-1)^{p^s-i}f_2(x)+(\gamma_0x-1)^{p^s+t-2i}f_1(x)h^{'}(x)\\
\equiv &\,(\gamma_0x-1)^{p^s-i}f_2(x)+(\gamma_0x-1)^th^{'}(x)+(\gamma_0x-1)^{p^s+t-2i}r(x)h^{'}(x)\\
\equiv &\,(\gamma_0x-1)^{p^s-i}[f_2(x)+(\gamma_0x-1)^tr(x)h^{'}(x)]+(\gamma_0x-1)^th^{'}(x)\\
\equiv &\,(\gamma_0x-1)^{p^s-i}f_2^{'}(x)+(\gamma_0x-1)^{p^s+t-2i}f_1^{'}(x)h^{'}(x)\\
\equiv &\,(\gamma_0x-1)^{p^s-i}f_2^{''}(x)+(\gamma_0x-1)^{p^s+t-2i}f_1^{'}(x)h^{'}(x)(\textrm{mod}(\gamma_0x-1)^{p^s}).
\end{split}
\end{displaymath}
Therefore, we can write Equation (\ref{e1}) as
\begin{displaymath}
\begin{split}
&(\gamma_0x-1)^i+u(\gamma_0x-1)^th(x)\\
\equiv & \,[f_1^{'}(x)+uf_2^{''}(x)][(\gamma_0x-1)^{p^s-i}+u(\gamma_0x-1)^{p^s+t-2i}h^{'}(x)]\\
\equiv & \,(\gamma_0x-1)^i+u[(\gamma_0x-1)^{p^s-i}f_2^{''}(x)+(\gamma_0x-1)^th^{'}(x)](\textrm{mod}(\gamma_0x-1)^{p^s}),
\end{split}
\end{displaymath}
where $f_1^{'}(x)=(\gamma_0x-1)^{2i-p^s}\,,\,\deg f_2^{''}(x)\leqslant i-1\,.$
It means that $$(\gamma_0x-1)^th(x)\equiv (\gamma_0x-1)^{p^s-i}f_2^{''}(x)+(\gamma_0x-1)^th^{'}(x)(\textrm{mod}(\gamma_0x-1)^{p^s}), $$
i.e., $$(\gamma_0x-1)^t[h(x)-h^{'}(x)]\equiv (\gamma_0x-1)^{p^s-i}f_2^{''}(x)(\textrm{mod}(\gamma_0x-1)^{p^s}).$$
If $i+t\leqslant p^s\leqslant 2i$, then $p^s-i> t$. So
$$(\gamma_0x-1)^t\,[\,h(x)-h^{'}(x)\,]\,\equiv\, (\gamma_0x-1)^t(\gamma_0x-1)^{p^s-i-t}f_2^{''}(x)(\textrm{mod}(\gamma_0x-1)^{p^s}), $$
which yields that there exists $s(x)\in \mathbb{F}_{p^m}[x]$ such that
$$(\gamma_0x-1)^t[h(x)-h^{'}(x)-(\gamma_0x-1)^{p^s-i-t}f_2^{''}(x)]=s(x)(\gamma_0x-1)^{p^s}.$$
Hence,
\begin{equation}
h(x)-h^{'}(x)-(\gamma_0x-1)^{p^s-i-t}f_2^{''}(x)=s(x)(\gamma_0x-1)^{p^s-t}. \label{e2}
\end{equation}
Note that $\deg h(x)\leqslant i-t-1$, $\deg h^{'}(x)\leqslant i-t$, $\deg f_2^{''}(x)\leqslant i-1$ and we have the degree of the left of Equation (\ref{e2}) $\leqslant \max\left \{ i-t-1\,,\, i-t\,,\, p^s-t-1 \right \}=p^s-t-1$. Compare the the degrees of two sides of Equation (\ref{e2}) and we have
$s(x)=0$, which follows that $$h(x)-h^{'}(x)=(\gamma_0x-1)^{p^s-i-t}f_2^{''}(x), $$ i.e., $$(\gamma_0x-1)^{p^s-i-t}|(h(x)-h^{'}(x)).$$

Sufficiency.
In order to prove $C\subseteq C^{\perp_\sigma} $, we just need to prove
$$ (\gamma_0x-1)^i+u(\gamma_0x-1)^th(x)  \in \left \langle (\gamma_0x-1)^{p^s-i}+u(\gamma_0x-1)^{p^s+t-2i}h^{'}(x) \right \rangle.$$
That means we just need to find $ f_i(x)\in \mathbb{F}_{p^m}[x]\, ,\, i=1, 2\, $ such that
$$ (\gamma_0x-1)^i+u(\gamma_0x-1)^th(x)\equiv[\,f_1(x)+uf_2(x)\,][\,(\gamma_0x-1)^{p^s-i}+u(\gamma_0x-1)^{p^s+t-2i}h^{'}(x)\,](\textrm{mod}(\gamma_0x-1)^{p^s}).$$

(a) When $p^s \leqslant i+t$, let $$f_1(x)=(\gamma_0x-1)^{2i-p^s}\,, \,f_2(x)=(\gamma_0x-1)^{i-p^s+t}[h(x)-h^{'}(x)], $$ and we have
$$ (\gamma_0x-1)^i+u(\gamma_0x-1)^th(x)=[\,f_1(x)+uf_2(x)\,][\,(\gamma_0x-1)^{p^s-i}+u(\gamma_0x-1)^{p^s+t-2i}h^{'}(x)\,].$$ It follows that $C\subseteq C^{\perp_\sigma} $, i.e., $C$ is $\sigma$-self-orthogonal.

(b) When $i+t\leqslant p^s\leqslant 2i$, since $(\gamma_0x-1)^{p^s-i-t}\,|\,(h(x)-h^{'}(x))$, there exists $m(x)\in \mathbb{F}_{p^m}[x]$ such that $h(x)-h^{'}(x)=(\gamma_0x-1)^{p^s-i-t}m(x)$. Let $$f_1(x)=(\gamma_0x-1)^{2i-p^s}\,, \,f_2(x)=m(x).$$ Then $$(\gamma_0x-1)^i+u(\gamma_0x-1)^th(x)=[f_1(x)+uf_2(x)][(\gamma_0x-1)^{p^s-i}+u(\gamma_0x-1)^{p^s+t-2i}h^{'}(x)], $$ which implies that $C\subseteq C^{\perp_\sigma} $, i.e., $C$ is $\sigma$-self-orthogonal.
\end{proof}

\begin{remark} Under the conditions of Theorem \ref{th10}, there doesn't exist $\sigma$-self-dual codefor any given $\sigma \in Aut(R)$. In fact, assume that there exists a $\sigma$-self-dual code $C$, then $|\,C\,|=|\,C^{\perp_\sigma}\,|$. By Theorem \ref{th3}, we have $|\,C\,|=p^{2m(p^s-i)}$, $|\,C^{\perp_\sigma}\,|=p^{2mi}$, which means $p^{2m(p^s-i)}=p^{2mi}$. Hence, $2i=p^s$, which contradicts that $p$ is an odd prime.
\end{remark}
\begin{theorem}\label{th11}
Let $C$ be a $\gamma$-constacyclic code of length $p^s$ over $R$ and $$C=\left \langle \,(\gamma _0x-1)^i+u(\gamma _0x-1)^th(x) \,\right \rangle\subseteq \mathcal{R}_{\gamma}, $$ where $\gamma _0=\theta(\gamma _0^{-1})$, $h(x)$ is a unit where it can be represented $h(x)=\sum _{j=0}^{p^s-i-1}h_j(\gamma _0x-1)^j$, $h_j\in \mathbb F_{p^m}$, $h_0\neq 0$ and $i> \frac{p^{s}+t}{2}$. Then $C$ is $\sigma$-self-orthogonal if and only if one of the following holds:

(a) $p^s\leqslant i+t$;

(b) $p^s > i+t$ and $(\gamma_0x-1)^{p^s-i-t}\,|\,[h(x)-h^{'}(x)]$, where $$\widetilde{h}(x)=-\sum _{j=0}^{p^s-i-1}\varepsilon \theta(h_j)(-1)^{j+t-i}(\gamma_0x-1)^jx^{i-j-t}=\sum _{j}h_j^{'}(\gamma_0x-1)^j, $$
$$h^{'}(x)=\sum _{j=0}^{p^s-i-1}h_j^{'}(\gamma_0x-1)^j.$$
\end{theorem}
\begin{proof}
By Theorem \ref{th4} and $\gamma _0=\theta(\gamma _0^{-1})$,
\begin{displaymath}
C^{\perp_\sigma}=\left \langle (\gamma_0x-1)^{i-t}-u\sum _{j=0}^{p^s-i-1}\varepsilon \theta(h_j)(-1)^{j+t-i}(\gamma_0x-1)^jx^{i-j-t}, u(\gamma_0x-1)^{p^s-i} \right \rangle.
\end{displaymath}
Let $$\widetilde{h}(x)=-\sum _{j=0}^{p^s-i-1}\varepsilon \theta(h_j)(-1)^{j+t-i}(\gamma_0x-1)^jx^{i-j-t}=\sum _{j}h_j^{'}(\gamma_0x-1)^j, $$
$$h^{'}(x)=\sum _{j=0}^{p^s-i-1}h_j^{'}(\gamma_0x-1)^j.$$ Then $$C^{\perp_\sigma}=\left \langle \,(\gamma_0x-1)^{i-t}+uh^{'}(x)\,,\, u(\gamma_0x-1)^{p^s-i}\, \right \rangle \subseteq \mathcal{R}_\gamma.$$

Necessity. Because $C\subseteq C^{\perp_\sigma}\, $, there exist $f_1(x)+uf_2(x)\,,\, g_1(x)+ug_2(x)\,,\, f_i(x)\,, \,g_i(x)\in \mathbb{F}_{p^m}[x]\, ,\, i=1, 2 $ such that
\begin{equation}
\begin{split}
&(\gamma_0x-1)^i+u(\gamma_0x-1)^th(x)\\
\equiv &\,[f_1(x)+uf_2(x)][(\gamma_0x-1)^{i-t}+uh^{'}(x)]+[g_1(x)+ug_2(x)]u(\gamma_0x-1)^{p^s-i}\\\label{e3}
\equiv &\,f_1(x)(\gamma_0x-1)^{i-t}+u[f_2(x)(\gamma_0x-1)^{i-t}+f_1(x)h^{'}(x)+g_1(x)(\gamma_0x-1)^{p^s-i}](\textrm{mod}(\gamma_0x-1)^{p^s}).
\end{split}
\end{equation}

Next, we will prove that we can find $ f_i(x)\,, \,g_i(x)\in \mathbb{F}_{p^m}[x]\, ,\, i=1, 2\, $  satisfying Equation (\ref{e3}) and
$$f_1(x)=(\gamma_0x-1)^{t}\,,\,\deg f_2(x)\leqslant p^s-i+t+1\,,\,\deg g_1(x)\leqslant i-1\,,\,g_2(x)=0.$$

So $(\gamma_0x-1)^i \,\equiv \,(\gamma_0x-1)^{i-t}f_1(x)(\textrm{mod}(\gamma_0x-1)^{p^s})$, which means that there exists $r(x)\in \mathbb{F}_{p^m}[x]$ such that $$(\gamma_0x-1)^{i-t}f_1(x)-(\gamma_0x-1)^i=r(x)(\gamma_0x-1)^{p^s}.$$
Hence, $f_1(x)=(\gamma_0x-1)^t+(\gamma_0x-1)^{p^s-i+t}r(x)\,.$ Let
$$f_1^{'}(x)=(\gamma_0x-1)^t\,, \,g^{'}(x)=g_1(x)+(\gamma_0x-1)^tr(x)h^{'}(x)$$
and write $f_2(x)$, $g^{'}(x)$ as
$$f_2(x)=\sum _jd_j(\gamma_0x-1)^j\,, \,g^{'}(x)=\sum _je_j(\gamma_0x-1)^j\,, \,d_j, e_j\in \mathbb F_{p^m}.$$
Let $f_2^{'}(x)=\sum _{j=0}^{p^s-i+t+1}d_j(\gamma_0x-1)^j$, $g^{''}(x)=\sum _{j=0}^{i-1}e_j(\gamma_0x-1)^j$. We have
\begin{displaymath}
\begin{split}
&f_2(x)(\gamma_0x-1)^{i-t}+f_1(x)h^{'}(x)+g_1(x)(\gamma_0x-1)^{p^s-i}\\
\equiv &\,f_2(x)(\gamma_0x-1)^{i-t}+(\gamma_0x-1)^th^{'}(x)+(\gamma_0x-1)^{p^s-i+t}r(x)h^{'}(x)+g_1(x)(\gamma_0x-1)^{p^s-i}\\
\equiv &\,f_2(x)(\gamma_0x-1)^{i-t}+f_1^{'}(x)h^{'}(x)+g^{'}(x)(\gamma_0x-1)^{p^s-i}\\
\equiv &\,f_2^{'}(x)(\gamma_0x-1)^{i-t}+f_1^{'}(x)h^{'}(x)+g^{''}(x)(\gamma_0x-1)^{p^s-i}(\textrm{mod}(\gamma_0x-1)^{p^s}).
\end{split}
\end{displaymath}
Therefore, we can write Equation (\ref{e3}) as
\begin{displaymath}
\begin{split}
&(\gamma_0x-1)^i+u(\gamma_0x-1)^th(x)\\
\equiv &\,[f_1^{'}(x)+uf_2^{'}(x)][(\gamma_0x-1)^{i-t}+uh^{'}(x)]+ug^{''}(x)(\gamma_0x-1)^{p^s-i}\\
\equiv &\,(\gamma_0x-1)^i+u[f_2^{'}(x)(\gamma_0x-1)^{i-t}+(\gamma_0x-1)^th^{'}(x)+g^{''}(x)(\gamma_0x-1)^{p^s-i}](\textrm{mod}(\gamma_0x-1)^{p^s}).
\end{split}
\end{displaymath}
where $f_1^{'}(x)=(\gamma_0x-1)^{t}\,,\,\deg f_2^{'}(x)\leqslant p^s-i+t+1\,,\,\deg g^{''}(x)\leqslant i-1\,.$
It means that $$(\gamma_0x-1)^th(x) \,\equiv\, f_2^{'}(x)(\gamma_0x-1)^{i-t}+(\gamma_0x-1)^th^{'}(x)+g^{''}(x)(\gamma_0x-1)^{p^s-i}(\textrm{mod}(\gamma_0x-1)^{p^s}).$$
From $i> \frac{p^{s}+t}{2}$, $p^s-i< i-t$. It follows that
\begin{displaymath}
\begin{split}
(\gamma_0x-1)^t[\,h(x)-h^{'}(x)\,] &\equiv\, f_2^{'}(x)(\gamma_0x-1)^{i-t}+g^{''}(x)(\gamma_0x-1)^{p^s-i}\\
&\equiv\, (\gamma_0x-1)^{p^s-i}[g^{''}(x)+f_2^{'}(x)(\gamma_0x-1)^{2i-t-p^s}](\textrm{mod}(\gamma_0x-1)^{p^s}).
\end{split}
\end{displaymath}
If $p^s > i+t$, i.e., $p^s-i> t$, we have
$$(\gamma_0x-1)^t[\,h(x)-h^{'}(x)\,] \equiv \,(\gamma_0x-1)^t(\gamma_0x-1)^{p^s-i-t}[\,g^{''}(x)+f_2^{'}(x)(\gamma_0x-1)^{2i-t-p^s}\,](\textrm{mod}(\gamma_0x-1)^{p^s}), $$ which yields that there exist $s(x)\in \mathbb{F}_{p^m}[x]$ such that
$$(\gamma_0x-1)^t\left \{ h(x)-h^{'}(x)-(\gamma_0x-1)^{p^s-i-t}[\,g^{''}(x)+f_2^{'}(x)(\gamma_0x-1)^{2i-t-p^s}\,] \right \}=s(x)(\gamma_0x-1)^{p^s}. $$
Hence,\begin{equation}
h(x)-h^{'}(x)-(\gamma_0x-1)^{p^s-i-t}[\,g^{''}(x)+f_2^{'}(x)(\gamma_0x-1)^{2i-t-p^s}\,]=s(x)(\gamma_0x-1)^{p^s-t}. \label{e4}
\end{equation}
Note that $\deg h(x)\leqslant p^s-i-1$, $\deg h^{'}(x)\leqslant p^s-i-1$, $\deg f_2^{'}(x)\leqslant p^s-i+t+1$, $\deg g^{'}(x)\leqslant i-1$
and we have the degree of the left of Equation (\ref{e4}) $\leqslant \max\left \{ p^s-i-1, p^s-t-1 \right \}=p^s-t-1\,.$ Compare the the degrees of two sides of Equation (\ref{e4}) and we have $s(x)=0\,, $ which follows that
$$ h(x)-h^{'}(x)=(\gamma_0x-1)^{p^s-i-t}[g^{''}(x)+f_2^{'}(x)(\gamma_0x-1)^{2i-t-p^s}], $$
i.e.,
$$(\gamma_0x-1)^{p^s-i-t}|(h(x)-h^{'}(x)).$$

Sufficiency.
In order to prove $C\subseteq C^{\perp_\sigma} $, we just need to prove
$$(\gamma_0x-1)^i+u(\gamma_0x-1)^th(x)\in \left \langle (\gamma_0x-1)^{i-t}-uh^{'}(x), u(\gamma_0x-1)^{p^s-i} \right \rangle.$$
That means we just need to find $ f_i(x)\,, \,g_i(x)\in \mathbb{F}_{p^m}[x]\, ,\, i=1, 2\, $ such that
\begin{displaymath}
\begin{split}
&(\gamma_0x-1)^i+u(\gamma_0x-1)^th(x)\\
\equiv\,& (f_1(x)+uf_2(x))[\,(\gamma_0x-1)^{i-t}+uh^{'}(x)\,]+(g_1(x)+ug_2(x))u(\gamma_0x-1)^{p^s-i}(\textrm{mod}(\gamma_0x-1)^{p^s}).
\end{split}
\end{displaymath}

(a) When $p^s \leqslant i+t$, i.e., $p^s-i\leqslant  t$, let
$$f_1(x)=(\gamma_0x-1)^t\,, \,g_1(x)=(\gamma_0x-1)^{t+i-p^s}[\,h(x)-h^{'}(x)\,], $$
and we have $$(\gamma_0x-1)^i+u(\gamma_0x-1)^th(x)\equiv f_1(x)[\,(\gamma_0x-1)^{i-t}+uh^{'}(x)\,]+g_1(x)u(\gamma_0x-1)^{p^s-i}(\textrm{mod}(\gamma_0x-1)^{p^s}).$$ It follows that $C\subseteq C^{\perp_\sigma} $, i.e., $C$ is $\sigma$-self-orthogonal.

(b) When $p^s > i+t$, since $(\gamma_0x-1)^{p^s-i-t}\,|\,(h(x)-h^{'}(x))$, there exists $m(x)\in \mathbb{F}_{p^m}[x]$ such that $h(x)-h^{'}(x)=(\gamma_0x-1)^{p^s-i-t}m(x)$. Let$$f_1(x)=(\gamma_0x-1)^t\,, \,g_1(x)=m(x)\,, $$ and we have
$$(\gamma_0x-1)^i+u(\gamma_0x-1)^th(x)\,\equiv\, f_1(x)[\,(\gamma_0x-1)^{i-t}+uh^{'}(x)\,]+g_1(x)u(\gamma_0x-1)^{p^s-i}(\textrm{mod}(\gamma_0x-1)^{p^s}), $$ which implies that $C\subseteq C^{\perp_\sigma} $, i.e., $C$  is $\sigma$-self-orthogonal.

\end{proof}

\begin{remark} Under the conditions of Theorem \ref{th11}, there doesn't exist $\sigma$-self-dual code for any given $\sigma \in Aut(R)$. In fact, if  $C$ is a Type 3 code, then $C^{\perp_\sigma} $ is a Type 4 code by Theorem \ref{th4}. Therefore, $C$ is not $\sigma$-self-dual.
\end{remark}

\begin{theorem}\label{th12}
Let $C$ be a $\gamma$-constacyclic code of length $p^s$ over $R$ and $$C=\left \langle \,(\gamma _0x-1)^i,  u(\gamma _0x-1)^\omega \,\right \rangle\subseteq \mathcal{R}_{\gamma},$$ where $\gamma _0=\theta(\gamma _0^{-1})$, $1\leqslant i\leqslant p^s-1$, $\omega < T\,.$ Then

(1) $C$ is $\sigma$-self-orthogonal if and only if $\omega +i\geqslant  p^s$,

(2) $C$ is $\sigma$-self-dual if and only if and $\omega +i= p^s$.
\end{theorem}
\begin{proof}
(1) By Theorem \ref{th4} and $\gamma _0=\theta(\gamma _0^{-1})$,
$$C^{\perp_\sigma}=\left \langle \,(\gamma_0x-1)^{p^s-\omega }, u (\gamma_0x-1)^{p^s-i }\,\right \rangle\subseteq \mathcal{R}_{\gamma}.$$
Necessity. Since $C\subseteq C^{\perp_\sigma} $, $(\gamma_0x-1)^{p^s-\omega }\,|\,(\gamma_0x-1)^{i}$, which follows that $p^s-\omega\leqslant i$, i.e., $\omega +i\geqslant  p^s\,.$

Sufficiency. If $\omega +i\geqslant  p^s$, then $p^s-\omega\leqslant i$, $p^s-i\leqslant \omega$, which means that $(\gamma_0x-1)^{i}\in C$, $u(\gamma_0x-1)^\omega\in C\,.$ Therefore, $C\subseteq C^{\perp_\sigma}\, $, i.e., $C$ is $\sigma$-self-orthogonal.

(2)
Necessity. Because $C = C^{\perp_\sigma} $, $p^s-\omega=i$, $\omega=p^s-i$, i.e., $\omega +i= p^s\,.$

Sufficiency. When $\omega +i= p^s$, $$C=\left \langle (\gamma_0x-1)^i,  u(\gamma_0x-1)^\omega \right \rangle\subseteq \mathcal{R}_{\gamma}, $$
$$C^{\perp_\sigma}=\left \langle (\gamma_0x-1)^{p^s-\omega }, u (\gamma_0x-1)^{p^s-i }\right \rangle=\left \langle (\gamma_0x-1)^i,  u(\gamma_0x-1)^\omega \right \rangle \subseteq \mathcal{R}_{\gamma}.$$
It implies that $C = C^{\perp_\sigma} $, i.e., $C$ is $\sigma$-self-dual.
\end{proof}

\begin{theorem}\label{th13}
Let $C$ be a $\gamma$-constacyclic code of length $p^s$ over $R$ and
$$C=\left \langle (\gamma _0x-1)^i+u(\gamma _0x-1)^th(x),  u(\gamma _0x-1)^\omega \right \rangle\subseteq \mathcal{R}_{\gamma}, $$   where $\gamma _0=\theta(\gamma _0^{-1})$, $1\leqslant i\leqslant p^s-1$, $\omega < T$, $h(x)$is a unit where it can be represented as $h(x)=\sum _{j=0}^{\omega-t-1}h_j(\gamma _0x-1)^j$, $h_j\in \mathbb F_{p^m}$, $h_0\neq 0$. Then $C$ is $\sigma$-self-orthogonal if and only if one of the following holds:

(a) $p^s \leqslant i+t$,

(b) $i+t< p^s\leqslant i+\omega $ and $(\gamma_0x-1)^{p^s-i-t}\,|\,[\,h(x)-h^{'}(x)\,]$,  where $$h^{'}(x)=-\sum _{j=0}^{\omega-t-1}h_j(-\gamma )^{i-t+j}(\gamma_0x-1)^jx^{i-j-t}.$$
\end{theorem}
\begin{proof}
By Theorem \ref{th4} and $\gamma _0=\theta(\gamma _0^{-1})$,
\begin{displaymath}
C^{\perp_\sigma}=\left \langle (\gamma_0x-1)^{p^s-\omega }-u(\gamma_0x-1)^{p^s-i-\omega+t }\sum _{j=0}^{\omega -t-1}\varepsilon \theta(h_j)(-1)^{j+t-i}(\gamma_0x-1)^jx^{i-j-t}, u (\gamma_0x-1)^{p^s-i }\right \rangle.
\end{displaymath}
Let $h^{'}(x)=-\sum _{j=0}^{\omega -t-1}\varepsilon \theta(h_j)(-1)^{j+t-i}(\gamma_0x-1)^jx^{i-j-t}.$ Then
$$C^{\perp_\sigma}=\left \langle (\gamma_0x-1)^{p^s-\omega }+u(\gamma_0x-1)^{p^s-i-\omega+t }h^{'}(x), u (\gamma_0x-1)^{p^s-i }\right \rangle\subseteq \mathcal{R}_{\gamma}.$$

Necessity. $C\subseteq C^{\perp_\sigma} $ implies that $(\gamma_0x-1)^i+u(\gamma_0x-1)^th(x)\in C^{\perp_\sigma}\,.$ It means that there exist $f_1(x)+uf_2(x)\,, \,g_1(x)+ug_2(x)\,, \,f_i(x)\,,\, g_i(x)\in \mathbb{F}_{p^m}[x] \,,\, i=1, 2 $  such that
\begin{equation}
\begin{split}
&(\gamma_0x-1)^i+u(\gamma_0x-1)^th(x)\\
\equiv &\,[f_1(x)+uf_2(x)][(\gamma_0x-1)^{p^s-\omega }+u(\gamma_0x-1)^{p^s-i-\omega+t }h^{'}(x)]+[g_1(x)+ug_2(x)]u(\gamma_0x-1)^{p^s-i}\\
\equiv &\,f_1(x)(\gamma_0x-1)^{p^s-\omega}+u[f_2(x)(\gamma_0x-1)^{p^s-\omega}+f_1(x)(\gamma_0x-1)^{p^s-i-\omega+t }h^{'}(x)+g_1(x)(\gamma_0x-1)^{p^s-i}]\\
&(\textrm{mod}(\gamma_0x-1)^{p^s}).\label{e5}
\end{split}
\end{equation}
Hence, $(\gamma_0x-1)^i \equiv \,(\gamma_0x-1)^{p^s-\omega}f_1(x)(\textrm{mod}(\gamma_0x-1)^{p^s})$ which implies that $r(x)\in \mathbb{F}_{p^m}[x]$ such that $$(\gamma_0x-1)^{p^s-\omega}f_1(x)-(\gamma_0x-1)^i=r(x)(\gamma_0x-1)^{p^s}.$$

Here, we will prove that $p^s-\omega\leqslant i$. Suppose $p^s-\omega>i$, i.e., $p^s-\omega-i> 0$, then we have
$$(\gamma_0x-1)^{p^s-\omega-i}f_1(x)-1=r(x)(\gamma_0x-1)^{p^s-i}, $$ i.e., $$(\gamma_0x-1)^{p^s-\omega-i}[\,f_1(x)-r(x)(\gamma_0x-1)^{\omega}\,]=1, $$ which yields that $(\gamma_0x-1)^{p^s-\omega-i}\,|\,1$, proving $p^s-\omega-i=0$. That contradicts that $p^s-\omega-i> 0$. Thus, $p^s-\omega\leqslant i$.

Next, we will find $ f_i(x)\,, \,g_i(x)\in \mathbb{F}_{p^m}[x]\, ,\, i=1, 2\, $  satisfying Equation (\ref{e5}) and
$$f_1(x)=(\gamma_0x-1)^{i+\omega -p^s}\,,\,\deg f_2(x)\leqslant \omega -1\,,\,\deg g_1(x)\leqslant i-1\,,\,g_2(x)=0.$$

So $f_1(x)=(\gamma_0x-1)^{i+\omega -p^s}+(\gamma_0x-1)^{\omega}r(x)\,.$  Let
$$f_1^{'}(x)=(\gamma_0x-1)^{i+\omega -p^s}\,, \,g^{'}(x)=g_1(x)+(\gamma_0x-1)^tr(x)h^{'}(x), $$
Write $f_2(x)$, $g^{'}(x)$ as
$$f_2(x)=\sum _jd_j(\gamma_0x-1)^j\,, \,g^{'}(x)=\sum _je_j(\gamma_0x-1)^j\,, \,d_j, e_j\in \mathbb F_{p^m}.$$
Let $$f_2^{'}(x)=\sum _{j=0}^{\omega -1}d_j(\gamma_0x-1)^j\,, \,g^{''}(x)=\sum _{j=0}^{i-1}e_j(\gamma_0x-1)^j\,.$$
Therefore, we can write Equation (\ref{e5}) as
\begin{displaymath}
\begin{split}
&(\gamma_0x-1)^i+u(\gamma_0x-1)^th(x)\\
\equiv &\,f_1(x)(\gamma_0x-1)^{p^s-\omega}+u[f_2(x)(\gamma_0x-1)^{p^s-\omega}+f_1(x)(\gamma_0x-1)^{p^s-i-\omega+t }h^{'}(x)+g_1(x)(\gamma_0x-1)^{p^s-i}]\\
\equiv &\,f_1^{'}(x)(\gamma_0x-1)^{p^s-\omega}+u[f_2^{'}(x)(\gamma_0x-1)^{p^s-\omega}+f_1^{'}(x)(\gamma_0x-1)^{p^s-i-\omega+t }h^{'}(x)+g^{''}(x)(\gamma_0x-1)^{p^s-i}]\\
\equiv &\,[f_1^{'}(x)+uf_2^{'}(x)][(\gamma_0x-1)^{p^s-\omega }+u(\gamma_0x-1)^{p^s-i-\omega+t }h^{'}(x)]+ug^{''}(x)(\gamma_0x-1)^{p^s-i}\\
\equiv &\,(\gamma_0x-1)^i+u[f_2^{'}(x)(\gamma_0x-1)^{p^s-\omega}+(\gamma_0x-1)^th^{'}(x)+g^{''}(x)(\gamma_0x-1)^{p^s-i}](\textrm{mod}(\gamma_0x-1)^{p^s}).
\end{split}
\end{displaymath}
 where $f_1^{'}(x)=(\gamma_0x-1)^{i+\omega -p^s}\,,\,\deg f_2^{'}(x)\leqslant \omega -1\,,\,\deg g^{''}(x)\leqslant i-1\,.$
It is obvious that $$(\gamma_0x-1)^th(x) \equiv f_2^{'}(x)(\gamma_0x-1)^{p^s-\omega}+(\gamma_0x-1)^th^{'}(x)+g^{''}(x)(\gamma_0x-1)^{p^s-i}(\textrm{mod}(\gamma_0x-1)^{p^s}, $$
which implies that $$(\gamma_0x-1)^t[h(x)-h^{'}(x)] \equiv (\gamma_0x-1)^{p^s-i}[f_2^{'}(x)(\gamma_0x-1)^{i-\omega}+g^{''}(x)](\textrm{mod}(\gamma_0x-1)^{p^s}).$$
It follows that there exists $s(x)\in \mathbb{F}_{p^m}[x]$ such that $$(\gamma_0x-1)^t[h(x)-h^{'}(x)]-(\gamma_0x-1)^{p^s-i}[f_2^{'}(x)(\gamma_0x-1)^{i-\omega}+g^{''}(x)]=s(x)(\gamma_0x-1)^{p^s}.$$
If $p^s> i+t$, then $p^s-i> t$. We have
\begin{equation}
h(x)-h^{'}(x)-(\gamma_0x-1)^{p^s-i-t}[f_2^{'}(x)(\gamma_0x-1)^{i-\omega}+g^{''}(x)]=s(x)(\gamma_0x-1)^{p^s-t}.\label{e6}
\end{equation}
Note that $\deg h(x)\leqslant \omega -t-1< p^s-t$, $\deg h^{'}(x)\leqslant i-t< p^s-t $, $\deg f_2^{'}(x)\leqslant \omega -1$, $\deg g^{''}(x)\leqslant i-1$, and we have the degree of the left of Equation (\ref{e6}) $< p^s-t$. Compare the the degrees of two sides of Equation (\ref{e6}) and we have $s(x)=0\,,$ which follows that $$h(x)-h^{'}(x)=(\gamma_0x-1)^{p^s-i-t}[\,f_2^{'}(x)(\gamma_0x-1)^{i-\omega}+g^{''}(x)\,].$$ This leads to $(\gamma_0x-1)^{p^s-i-t}\,|\,[\,h(x)-h^{'}(x)\,]$.

Sufficiency.
Since $p^s\leqslant i+\omega $, i.e., $p^s-i\leqslant \omega $, $u(\gamma_0x-1)^\omega\in C^{\perp_\sigma} $.

(a) If $p^s\leqslant i+t $, let $g_1(x)=(\gamma_0x-1)^{t+i-p^s}$, $f_1(x)=(\gamma_0x-1)^{\omega +i-p^s}$ and we have
\begin{displaymath}
\begin{split}
&(\gamma_0x-1)^i+u(\gamma_0x-1)^th(x)\\
\equiv\,& f_1(x)[(\gamma_0x-1)^{p^s-\omega }+u(\gamma_0x-1)^{p^s-i-\omega+t }h^{'}(x)]+g_1(x)u(\gamma_0x-1)^{p^s-i}(\textrm{mod}(\gamma_0x-1)^{p^s}).
\end{split}
\end{displaymath}
 It follows that $C\subseteq C^{\perp_\sigma} $, i.e., $C$ is $\sigma$-self-orthogonal.

(b) If $i+t<  p^s\leqslant i+\omega $, since $(\gamma_0x-1)^{p^s-i-t}|h(x)-h^{'}(x)$, there exists $m(x)\in \mathbb{F}_{p^m}[x]$ such that $$h(x)-h^{'}(x)=(\gamma_0x-1)^{p^s-i-t}m(x).$$  Let $f_1(x)=(\gamma_0x-1)^{\omega +i-p^s}$, $g_1(x)=m(x)$ and we have
\begin{displaymath}
\begin{split}
&(\gamma_0x-1)^i+u(\gamma_0x-1)^th(x)\\
\equiv &\,f_1(x)[(\gamma_0x-1)^{p^s-\omega }+u(\gamma_0x-1)^{p^s-i-\omega+t }h^{'}(x)]+g_1(x)u(\gamma_0x-1)^{p^s-i}(\textrm{mod}(\gamma_0x-1)^{p^s}),
\end{split}
\end{displaymath}
 which implies that $C\subseteq C^{\perp_\sigma} $, i.e., $C$  is $\sigma$-self-orthogonal.
\end{proof}

\begin{corollary}
Under the same conditions as Theorem \ref{th13}, then $C$ is $\sigma$-self-dual if and only if $\omega +i= p^s$, and $(\gamma_0x-1)^{p^s-i-t}\,|\,[\,h(x)-h^{'}(x)\,]$, where $h^{'}(x)$ is same as Theorem \ref{th13}.
\end{corollary}
\begin{proof}
Necessity. If $C$ is $\sigma$-self-dual, i.e., $C = C^{\perp_\sigma} $, then $|\,C\,|=|\,C^{\perp_\sigma}\,|$. It follows from Theorem \ref{th3} that $|\,C\,|=p^{m(2p^s-i-\omega )}$, $|\,C^{\perp_\sigma}\,|=p^{m(i+\omega )}$. So $p^{m(2p^s-i-\omega )}=p^{m(i+\omega )}$, which implies that $p^s=i+\omega $. Since $C$ is $\sigma$-self-orthogonal, we have $(\gamma_0x-1)^{p^s-i}|h(x)-h^{'}(x)$.

Sufficiency. If $p^s=i+\omega $ and $(\gamma_0x-1)^{p^s-i-t}\,|\,(h(x)-h^{'}(x))$, we have $C$ is $\sigma$-self-orthogonal by Theorem 4.9
 and $$C=\left \langle \,(\gamma_0x-1)^i+u(\gamma_0x-1)^th(x),  u(\gamma_0x-1)^\omega \,\right \rangle\subseteq \mathcal{R}_{\gamma}, $$
$$C^{\perp_\sigma}=\left \langle \,(\gamma_0x-1)^i+u(\gamma_0x-1)^th^{'}(x),  u(\gamma_0x-1)^\omega\, \right \rangle\subseteq \mathcal{R}_{\gamma}.$$
$(\gamma_0x-1)^{p^s-i-t}\,|\,(h(x)-h^{'}(x))$ means that there exists $m(x)\in \mathbb{F}_{p^m}[x]$ such that
$$h(x)-h^{'}(x)=(\gamma_0x-1)^{p^s-i-t}m(x).$$
Hence, $$(\gamma_0x-1)^i+u(\gamma_0x-1)^th(x)=(\gamma_0x-1)^i+u(\gamma_0x-1)^th^{'}(x)+u(\gamma_0x-1)^\omega m(x), $$
Therefore, $$(\gamma_0x-1)^i+u(\gamma_0x-1)^th^{'}(x)=(\gamma_0x-1)^i+u(\gamma_0x-1)^th(x)-u(\gamma_0x-1)^\omega m(x), $$
which yields that $(\gamma_0x-1)^i+u(\gamma_0x-1)^th^{'}(x)\in C\,.$ Moreover, $u(\gamma_0x-1)^\omega\in C $, we have $C^{\perp_\sigma}\subseteq C\, $. Since $C$ is $\sigma$-self-orthogonal, i.e., $C\subseteq C^{\perp_\sigma}\, $. $C = C^{\perp_\sigma}$, which implies that $C$ is $\sigma$-self-dual.
\end{proof}

\begin{remark}
Particularly, if we take $\sigma=1$, then $\sigma$-self-orthogonal code is the usual self-orthogonal code. Then we can obtain the $\sigma$-self-orthogonality of the constacyclic codes from the above results.
\end{remark}

\begin{example}\label{eg1}
Let $\mathbb F_{3}=\{0,1,-1\}$, $\mathbb F_{3^2}=\frac{\mathbb F_{3}[x]}{\langle x^2+1\rangle}\cong\{0,1,-1,\omega,-\omega,1+\omega,-1+\omega,1-\omega,-1-\omega\}$, where $\omega^2+1=0$ and $R=\mathbb F_{3^2}+u\mathbb F_{3^2}$. We consider the code $C$ of length $3^2=9$ over $R$, whose generator matrix is
$$G_1=\left(
  \begin{array}{ccccccccc}
    1 & 0 & 1 & \omega & 0 & \omega&-1&0&-1 \\
    0&1&\omega&0&\omega&-1&0&-1&-\omega\\
    0&0&u&0&0&u\omega&0&0&-u\\
    0&0&0&u&-u\omega&-u&-u\omega&-u&u\omega
  \end{array}
\right).$$
Then $C$ is an $\omega$-constacyclic code, $C=\left \langle (\omega x+1)^7, u(\omega x+1)^5\right \rangle \subseteq R[x] / \langle x^9-\omega \rangle$. Then the size of $C$ is $81^3$.
Let $\sigma$ be an automorphism of $R$ defined by $\sigma:a+ub\longmapsto a^3+ub^3$. It is easy to verify that $C$ is $\sigma$-self-orthogonal but not self-orthogonal or $\sigma$-self-dual.
\end{example}

\begin{example}\label{eg2}
Let $\mathbb F_{5}=\{0,1,2,3,4\}$, $\mathbb F_{5^2}=\frac{\mathbb F_{5}[x]}{\langle x^2+3\rangle}\cong\{a+\omega b|a,b\in \mathbb F_{5},\,\omega^2=2\}$ and $R=\mathbb F_{3^2}+u\mathbb F_{3^2}$. We consider the code $C$ of length $5$ over $R$, whose generator matrix is
$$G_2=\left(
  \begin{array}{ccccccccc}
    1 & 2+2\omega & 2+3\omega+u(2+3\omega) & 1+3u & 2+2\omega+u(2+2\omega) \\
    0 & u & u(4+4\omega)&u(1+4\omega)&4u
  \end{array}
\right).$$
Then $C$ is an $2+2\omega$-constacyclic code, $C=\left \langle ((2+2\omega) x-1)^4+u((2+2\omega) x-1)^2, u((2+2\omega)x+1)^3\right \rangle \subseteq R[x] / \langle x^5-(2+2\omega) \rangle$. Then the size of $C$ is $25^3$.
Let $\sigma$ be an automorphism of $R$ defined by $\sigma:a+ub\longmapsto a^5+ub^5$. It is easy to verify that $C$ is $\sigma$-self-orthogonal but not self-orthogonal or $\sigma$-self-dual.
\end{example}

\begin{example}\label{eg3}
Let $R=\mathbb F_{3^2}+u\mathbb F_{3^2}$ as Example \ref{eg1}. We consider the code $C$ of length $3^2=9$ over $R$, whose generator matrix is
$$G_3=\left(
  \begin{array}{ccccccccc}
    1&\omega&-1&-\omega&1&\omega&-1&-\omega&1 \\
  \end{array}
\right).$$
Then $C$ is an $\omega$-constacyclic code, $C=\left \langle \,(\omega x+1)^8\,\right \rangle \subseteq R[x] / \langle x^9-1 \rangle$. Since the size of $C$ is $|C|=81$, the Hamming distance of $C$ is $d=9$, the length of $C$ is $n=9$ and the cardinality of the code alphabet is $|R|=81$, we have $|C|=|R|^{n-d+1}$, which means that $C$ is a MDS code.
Let $\sigma$ be an automorphism of $R$ defined by $\sigma:a+ub\longmapsto a^3+ub^3$. It is easy to find that $C$ is $\sigma$-self-orthogonal but not self-orthogonal. Thus $C$ is a $\sigma$-self-orthogonal MDS $\omega$-constacyclic code of length $9$ over $\mathbb F_{3^2}+u\mathbb F_{3^2}$.

\end{example}

\begin{remark}
 The results about $\sigma$-self-orthogonality of constacyclic codes of length $p^s$ over $R$ can be extended to constacyclic codes of length $2p^s$. When we consider constacyclic codes of length $2p^s$, the situation of $\lambda$ is divided into three cases separately:

(a) $\lambda$ is a square unit of $\mathbb F_{p^m}+u\mathbb F_{p^m}$,

(b) $\lambda =\alpha +u\beta$ is not a square and $\alpha \, $, $\beta $ are nonzero elements of $\mathbb F_{p^m}$,

(c) $\lambda =\gamma $ is not a square and $\gamma\, $ is a nonzero element of $\mathbb F_{p^m}$.

When $\lambda=\alpha^2$ is a square(i.e., case (a)), by Chinese Remainder Theorem, it is easy to find that every $\sigma$-self-orthogonal ($\sigma$-self-dual) $\lambda$-constacyclic code $C$ of length $2p^s$ can be represented as a direct sum of a $\sigma$-self-orthogonal ($\sigma$-self-dual) $-\alpha$-constacyclic code $C_1$ and a $\sigma$-self-orthogonal ($\sigma$-self-dual) $\alpha$-constacyclic code $C_2$ of length $p^s$ over $R$.

When $\lambda$ is not a square(i.e., case (b,c)), \cite{r15} gave the structures of $\lambda$-constacyclic codes and the dual codes of length $2p^s$, which are similar with those of length $p^s$. Thus, we can use the similar method to get the $\sigma$-self-orthogonality of the $\lambda$-constacyclic codes of length $2p^s$.

\end{remark}

\section{Acknowledgments}

We sincerely thank Professor Jay A. Wood for his valuable suggestions and comments during his visit at Central China Normal University from April to May in this year.  We also thank Professor Yun Fan for his comments. This work was supported by the self-determined research funds of CCNU from the colleges's basic research and operation of MOE (Grant NO.CCNU18TS028).



\begin{thebibliography}{99}
\addcontentsline{toc}{section}{\protect \heiti Reference} 
\bibitem{d20}Y. Alkhamees, The determination of the group of automorphisms of a finite chain ring of
characteristic $p$. \emph{The Quarterly Journal of Mathematics}, 1991, \textbf{42}, 387-391.

\bibitem{r12}C. Bachoc, Applications of coding theory to the construction of modular lattices. \emph{Journal of Combinatorial Theory},  Series A,  1997,  \textbf{78}(1): 92-119.

\bibitem{r2}E.R. Berlekamp, \emph{Algebraic Coding Theory}. Mc Graw-Hill Book Company, 1968.

\bibitem{r3}E.R. Berlekamp, Negacyclic codes for the Lee metric. \emph{Proceedings of the Conference on Combinatorial Mathematics and Its Applications},  1968: 298-316.

\bibitem{b25}T. Blackford, Cyclic codes over $\mathbb{Z}_4$ of oddly even length[J].  \emph{Discrete Applied Mathematics},  2003,  \textbf{128}(1): 27-46.

\bibitem{r4}A.R. Calderbank, A.R. Hammons, P.V. Kumar, N.J.A. Sloane, P. Sol\'{e}, A linear construction for certain Kerdock and Preparata codes. \emph{ Bulletin of the American Mathematical Society},  1993,  \textbf{29}(2): 218-222.

\bibitem{a17}A.R. Calderbank, E.M. Rains, P.W. Shor, N.J.A. Sloane, Quantum error correction via codes over $GF(4)$. \emph{IEEE Transactions on Information Theory}, 1998, \textbf{44}(4): 1369-1387.

\bibitem{r15}B. Chen, H.Q. Dinh, H. Liu, L. Wang, Constacyclic codes of length $2p^s$ over $\mathbb F_{p^m}+u\mathbb F_{p^m}$. \emph{Finite Fields and Their Applications},  2016,  \textbf{37}(1): 108-130.

\bibitem{r13}H.Q. Dinh, Constacyclic codes of length $2^s$ over Galois extension rings of $\mathbb F_2+u\mathbb F_2$. \emph{IEEE Transactions on Information Theory}, 2009, \textbf{55}(4): 1730-1740.

\bibitem{r14}H.Q. Dinh, Constacyclic codes of length $p^s$ over $\mathbb F_{p^m}+u\mathbb F_{p^m}$. \emph{Journal of Algebra},  2010,  \textbf{324}(5): 940-950.

\bibitem{r10}H.Q. Dinh, Negacyclic codes of length $2^s$ over Galois rings. \emph{IEEE Transactions on Information Theory},  2005,  \textbf{51}(12): 4252-4262.

\bibitem{r16}H.Q. Dinh, Y. Fan, H. Liu, X. Liu, S. Sriboonchitta, On self-dual constacyclic codes of length $p^s$ over $\mathbb F_{p^m}+u\mathbb F_{p^m}$. \emph{Discrete Mathematics}, 2017, \textbf{341}(2): 324-335.

\bibitem{r7}H.Q. Dinh, S.R. L\'{o}pez-Permouth, Cyclic and negacyclic codes over finite chain rings. \emph{IEEE Transactions on Information Theory},  2004,  \textbf{50}(8): 1728-1744.

\bibitem{b27}S.T. Dougherty, S. Ling, Cyclic codes over $\mathbb{Z}_4$ of even length[J].  \emph{Designs, Codes and Cryptography},  2006,  \textbf{39}(2): 127-153.

\bibitem{r8}S.T. Dougherty, Y.H. Park, On modular cyclic codes. \emph{Finite Fields and Their Applications}, 2007, \textbf{13}(1): 31-57.

\bibitem{c16}Y. Fan, L. Zhang, Galois self-dual constacyclic codes. \emph{Designs Codes and Cryptography}, 2017, \textbf{84}(3): 473¨C492.

\bibitem{r5}A.R. Hammons, P.V. Kumar, A.R. Calderbank, N.J.A. Sloane, P. Sol\'{e}, The $\mathbb{Z}_4$-linearity of Kerdock,  Preparata,  Goethals,  and related codes. \emph{IEEE Transactions on Information Theory},  1994,  \textbf{40}(2): 301-319.

\bibitem{r11}H. Liu, Y. Maouche, Some repeated-root constacyclic codes over Galois rings. \emph{IEEE Transactions on Information Theory}, 2017, \textbf{63}(10): 6247-6255.

\bibitem{r6}A.A. Nechaev, Kerdock code in a cyclic form. \emph{Discrete Mathematics and Applications},  1991,  \textbf{1}(4): 365-384.

\bibitem{r9}R. Sobhani, M. Esmaeili, Cyclic and negacyclic codes over the Galois ring $GR(p^2,m)$. \emph{Discrete Applied Mathematics}, 2009, \textbf{157}(13): 2892-2903.









\
\end{thebibliography}
\end{document}